\documentclass[runningheads]{llncs}
\usepackage[utf8]{inputenc}
\usepackage{ltl}
\usepackage{graphicx}
\usepackage{comment}

\usepackage{graphicx} 
\usepackage{stmaryrd}
\usepackage{lineno}
\usepackage{amssymb}
\usepackage{pifont}
\newcommand{\cmark}{\ding{51}}%
\newcommand{\xmark}{\ding{55}}%
\usepackage{algorithm}
\usepackage{algpseudocode}
\usepackage{mathtools}
\usepackage{wrapfig}
\usepackage[export]{adjustbox} 
\usepackage{caption} 
\usepackage{tikz}
\usepackage{pgfplots}
\pgfplotsset{compat=1.12}
\usepackage{placeins}
\usepackage{subcaption}
\usepackage{multicol}
\usetikzlibrary{matrix,shapes.geometric,arrows,calc,fit}
\usepackage{listings}
\usepackage{color}
\definecolor{deepblue}{rgb}{0,0,0.5}
\definecolor{deepred}{rgb}{0.6,0,0}
\definecolor{deepgreen}{rgb}{0,0.5,0}

\usepackage{listings}

 \usepackage{url}
 
\tikzset{ 
    table/.style={
        matrix of nodes,
        row sep=-\pgflinewidth,
        column sep=-\pgflinewidth,
        nodes={
            rectangle,
            draw=black,
            align=center
        },
        minimum height=1.5em,
        text depth=0.5ex,
        text height=2ex,
        nodes in empty cells,
        every even row/.style={
            nodes={fill=gray!20}
        },
        column 1/.style={
            nodes={text width=2em,font=\bfseries}
        },
        column 3/.style={
            nodes={text width=40em,font=\bfseries}
        },
        row 1/.style={
            nodes={
                fill=black,
                text=white,
                font=\bfseries
            }
        }
    }
}

\newcommand\ap{\text{AP}}
\newcommand\Un{\mathcal{U}}
\newcommand\Leaf{\ensuremath\mathsf{Leaf}}
\newcommand\Fallback{\ensuremath\mathsf{Fback}}
\newcommand\Sequence{\ensuremath\mathsf{Seq}}
\newcommand\Parallel{\ensuremath\mathsf{Par}}

\newcommand\F{
\protect\tikz[baseline]{
    \protect\path[draw,line width=.12ex,line join=round]
      (0ex,.6ex) -- (.95ex,1.55ex) -- (1.9ex,.6ex) -- (.95ex,-.35ex) -- cycle;
  }}

\newcommand\G{
  \protect\tikz[baseline]{
    \draw[line width=.12ex,line join=round]
      (0ex,-.2ex) -- (0ex,1.3ex) -- (1.5ex,1.3ex) -- (1.5ex,.-.2ex) -- cycle;
  }}
\newcommand\X{\bigcirc}
\newcommand\Encode{\textsf{encode}}

\newcommand\Local{\textsf{local}}
\newcommand\Incremental{\textsf{incremental}}
\newcommand\Landmark{\textsf{landmark-based}}
\newcommand\Repair{\textsf{repair}_{\sigma_R}}

\newcommand\tree{\mathcal{T}}

\newcommand\true{\texttt{true}}
\newcommand\false{\texttt{false}}

\newcommand\found{\mathit{found}}
\newcommand\atLfirst{\mathit{atL_1}}
\newcommand\atLsecond{\mathit{atL_2}}
\newcommand\atLthird{\mathit{atL_3}}
\newcommand\avoidArea{\mathit{avoidArea}}
\newcommand\hasKey{\mathit{hasKey}}
\newcommand\openDoor{\mathit{openDoor}}
\newcommand\nextTo{\mathit{diagonalBehind}}
\newcommand\alignedHeading{\mathit{alignedHeading}}
\newcommand\aboveTouchdown{\mathit{aboveTouchdown}}
\newcommand\descended{\mathit{descended}}
\newcommand\stayPos{\mathit{stayPos}}
\newcommand\landing{\mathit{land}}
\newcommand\valid{\mathit{valid}}
\newcommand\loose{\mathit{loose}}

\newcommand\candidates{\mathit{candidates}}
\newcommand\landmarks{\mathit{landmark}}
\newcommand\computeCandidates{\textsf{computeCandidates}}

\usepackage{xcolor-material}
\tikzstyle{startstop} = [rectangle, rounded corners, text centered, draw=black, fill=red!30]
\tikzstyle{process} = [rectangle, text centered, draw=black, fill=blue!30]
\tikzstyle{decision} = [ellipse, text centered, draw=black, fill=green!30]
\tikzstyle{incrementalrepair} = [dotted, thick, draw=blue, rounded corners]
\tikzstyle{landmarkrepair} = [dotted, thick, draw=orange, rounded corners]
\tikzstyle{milpencoding} = [dotted, thick, draw=black, rounded corners]
\tikzstyle{arrow} = [thick,->,>=stealth]
\tikzset{%
   apple/.pic={
  \begin{scope}[scale=0.17] 
    \fill [MaterialBrown] (-1/8,0) 
      arc (180:120:1 and 3/2) coordinate [pos=3/5] (@)-- ++(1/6,-1/7) 
      arc (120:180:5/4 and 3/2) -- cycle;
    \fill [MaterialLightGreen500] (0,-9/10) 
      .. controls ++(180:1/8) and ++(  0:1/4) .. (-1/3,  -1)
      .. controls ++(180:1/3) and ++(270:1/2) .. (  -1,   0)
      .. controls ++( 90:1/3) and ++(180:1/3) .. (-1/2, 3/4)
      .. controls ++(  0:1/8) and ++(135:1/8) .. (   0, 4/7)
      .. controls ++( 45:1/8) and ++(180:1/8) .. ( 1/2, 3/4)
      .. controls ++(  0:1/3) and ++( 90:1/3) .. (   1,   0)
      .. controls ++(270:1/2) and ++(  0:1/3) .. ( 1/3,  -1)
      .. controls ++(180:1/4) and ++(  0:1/8) .. cycle;
    \fill [MaterialLightGreen600] (0, 4/7)
      .. controls ++( 45:1/8) and ++(180:1/8) .. ( 1/2, 3/4)
      .. controls ++(  0:1/3) and ++( 90:1/3) .. (   1,   0)
      .. controls ++(270:1/2) and ++(  0:1/3) .. ( 1/3,  -1)
      .. controls ++(180:1/4) and ++(  0:1/8) .. (   0,-9/10);
    \fill [MaterialGreen500, shift={(@)}, rotate=-30] 
      (0,0) arc (45:135:3/4 and 3/5) arc (225:315:3/4 and 3/5);
    \fill [MaterialGreen700, shift={(@)}, rotate=-30] 
      (0,0) arc (315:225:3/4 and 3/5) -- cycle;
  \end{scope} 
},
  orange/.pic={
  \begin{scope}[scale=0.17] 
        \fill [MaterialOrange500] (0,0) circle [radius=1];
        \fill [MaterialOrange600] (0,0) -- (45:1) arc (45:-135:1) -- cycle;
        \fill [MaterialOrange700, shift={(0,3/4)}] coordinate (@)
        ellipse [x radius=1/4, y radius=1/8];
        \begin{scope}
        \clip (0,0) circle [radius=1];
        \fill [MaterialOrange700, shift=(@)] (90:1/4 and 1/8) 
          \foreach \i [evaluate={\j=mod(\i,2)+1/4;}]in {0,...,12}{
          -- (90+\i*30:\j*3/4 and \j*3/8) } -- cycle;
        \end{scope}
        \fill [MaterialBrown] (-1/16, 3/4) -- ++(0,1/4) arc (180:0:1/16 and 1/32)
           -- ++(0,-1/4) arc (360:180:1/16 and 1/32) -- cycle;
        \fill [MaterialGreen500, shift=(@), rotate=-150] 
          (0,0) arc (45:135:1/2 and 4/5) arc (225:315:1/2 and 3/5);
        \fill [MaterialGreen700, shift=(@), rotate=-150] 
          (0,0) arc (45:135:1/2 and 4/5) -- cycle;
  \end{scope} 
},
}
\newcommand{\drawOrangeWithText}[3]{
    \begin{tikzpicture}
        \path pic {orange};
       \node at (#1, #2) {#3}; 
    \end{tikzpicture}
}
\newcommand{\drawAppleWithText}[3]{
    \begin{tikzpicture}
        \path pic {apple};
       \node at (#1, #2) {#3}; 
    \end{tikzpicture}
}
\newcommand{\drawOrange}[1]{
\hspace{-3mm}
\raisebox{#1\height}{
    \begin{tikzpicture}
            \path pic {orange};
    \end{tikzpicture}
    }
\hspace{-3mm}
}
\newcommand{\drawApple}[1]{
\hspace{-3mm}
\raisebox{#1\height}{
    \begin{tikzpicture}
            \path pic {apple};
    \end{tikzpicture}
    }
\hspace{-3mm}
}

\begin{document}
%
\title{Trace Repair for Temporal Behavior Trees}

\author{Sebastian Schirmer\inst{1}\orcidID{0000-0002-4596-2479} \and
Philipp Schitz\inst{1}\orcidID{0000-0001-7365-3430} \and
Johann C.\ Dauer\inst{1}\orcidID{0000-0002-8287-2376}\and
Bernd Finkbeiner\inst{2}\orcidID{0000-0002-4280-8441}\and
Sriram Sankaranarayanan\inst{3}\orcidID{0000-0001-7315-4340}
}
\authorrunning{Schirmer et al.}
%
\institute{DLR German Aerospace Center, Germany, \email{firstname.lastname@dlr.de}\and
CISPA Helmholtz Center for Information Security, Germany, \email{finkbeiner@cispa.de}\and
University of Colorado Boulder, USA, \email{srirams@colorado.edu}}

\maketitle              

\begin{abstract}
We present methods for repairing traces against specifications given as temporal behavior trees (TBT). 
TBT are a specification formalism for action sequences in robotics and cyber-physical systems, where specifications of sub-behaviors, given in signal temporal logic, are composed using operators for sequential and parallel composition, fallbacks, and repetition. 
Trace repairs are useful to explain failures and as training examples that avoid the observed problems. 
In principle, repairs can be obtained via mixed-integer linear programming (MILP), but this is far too expensive for practical applications. 
We present two practical repair strategies: (1) incremental repair, which reduces the MILP by splitting the trace into segments, and (2) landmark-based repair, which solves the repair problem iteratively using TBT's robust semantics as a heuristic that approximates MILP with more efficient linear programming. 
In our experiments, we were able to repair traces with more than $25,000$ entries in under ten minutes, while MILP runs out of memory.
\end{abstract}

\section{Introduction}
We study the problem of trace repair for temporal behavior tree specifications.
Behavior trees specify complex  sequences of actions for applications in robotics and cyber-physical systems (CPS)  using operators for sequential composition, fallbacks, repetitions, and parallel executions. Specification formalisms based on behavior trees  have been widely studied for applications such as runtime monitoring of CPS~\cite{6606326,10.1162/ARTL_a_00192,DBLP:conf/sle/GhzouliBJDW20,10155191,9921558,9355019}. \emph{Temporal behavior trees} (TBTs) are one such formalism based on behavior trees~\cite{10.1145/3641513.3650180}. TBTs combine formulas in Signal Temporal Logic (STL) using the operators defined by behavior trees. 
They have been shown to be strictly more powerful than linear temporal logic and equivalent to formalisms such as regular temporal logic~\cite{leucker2007regular}. 
Previous work has studied the construction of monitors that check whether a given trace satisfies a TBT specification, and provides \emph{quantitative robustness semantics} that measures the distance between a trace and a TBT~\cite{10.1145/3641513.3650180}. 
Furthermore, the monitor is able to 
split the input trace into segments while mapping each segment to a subtree of the TBT specification so that the overall Boolean verdict of the monitor can be justified based on the segmentation. 

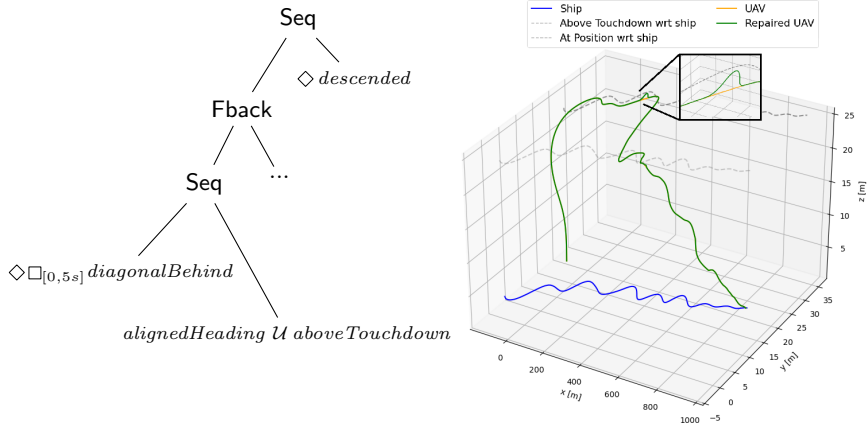
\begin{figure}[b!]
\begin{minipage}{0.45\textwidth}
    \begin{tikzpicture}
        \node {$\Sequence$}[sibling distance=15mm, level distance=12mm]
            child { 
            node[align=center] {
            $\Fallback$}[sibling distance=10mm, level distance=8mm]
                child[level distance=10mm] {node {
                $\Sequence$\\}[sibling distance=22mm, level distance=10mm]
                    child[level distance=12mm] {
                        node {\scriptsize $\F\,\G_{[0,5s]} \nextTo$}}
                    child[level distance=20mm] {
                        node {\scriptsize $\alignedHeading\ \Un\, \aboveTouchdown$ }}
                }
                child[level distance=9mm] {node {...}}
            }
            child[level distance=8mm] {node {\scriptsize $\F\, \descended$ }};
    \end{tikzpicture}
    \vspace{6mm}
    \subcaption{Excerpt from a TBT specifying a UAV landing maneuver on a ship.}\label{fig:tbt}
\end{minipage}
\hfill
\begin{minipage}{0.5\textwidth}
    \includegraphics[scale=0.3]{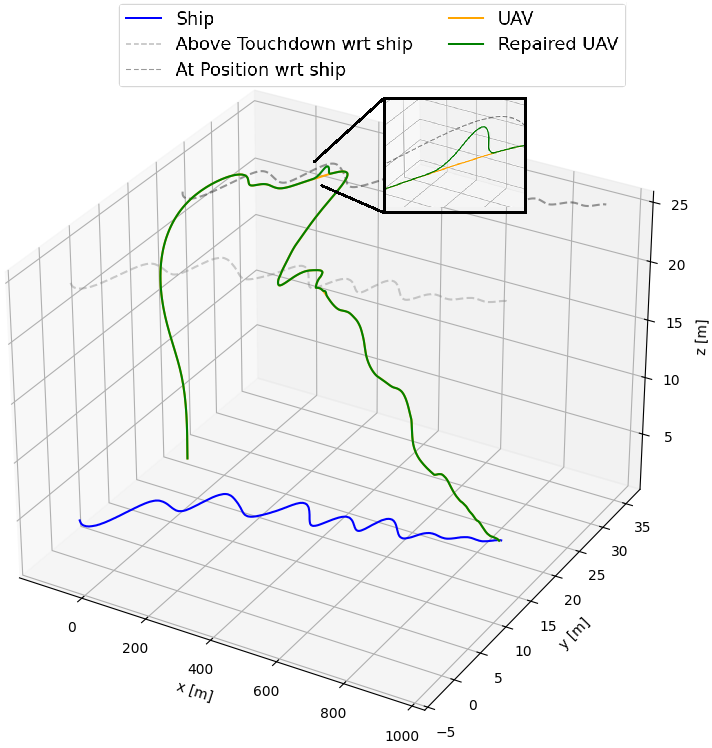}
    \subcaption{Trace repair of an execution of the landing maneuver against the TBT in Figure \ref{fig:tbt}.}\label{fig:plotlanding}
\end{minipage}
\caption{Repair of a UAV ship landing maneuver~\cite{10.1145/3641513.3650180} against a TBT specification.}
\label{fig:introduction}
\end{figure}

In this paper, we study the case that the observed behavior does \emph{not} satisfy the specification. 
We are interested in solving the \emph{trace repair} problem: 
given a violating trace and a specification, compute a trace that satisfies the specification while minimally modifying the original trace.
The repaired trace can be used as a modified plan that corrects the output of an untrusted planner based on a correctness specification. Furthermore, it provides valuable information about what went wrong in the violating trace. 
For systems that involve a machine-learned component (or a human operator), they are also useful as training examples to avoid the problem in the future.

Figure~\ref{fig:tbt} shows an excerpt of a TBT that specifies an automated landing maneuver of an unmanned aerial vehicle (UAV) on a ship.
The TBT decomposes the maneuver into two parts: first, the UAV moves to a position diagonally behind the ship and holds there for five seconds; second, the UAV transitions above the touchdown point while aligning its heading with the heading of the ship. Typically, the TBT would include other fallback maneuvers, such as approaching the ship from behind or from the side, which we omit here for brevity.
The formulas at the TBT's leaf nodes specify the successful execution of the various actions.
For instance, the UAV reaching and holding a diagonal position behind the ship is expressed by the STL formula  $\F\,\G_{[0,5s]} \nextTo$.

Figure~\ref{fig:plotlanding} shows a plot of a landing maneuver. Given a trace of the system events, a TBT monitor segments the trace by assigning parts of the trace to corresponding TBT nodes~\cite{10.1145/3641513.3650180}.
The monitor identifies which parts of the TBT are satisfied and which are violated. However, this information alone is still often difficult to interpret, because especially minor deviations that violate the TBT are difficult to detect manually.
In Figure \ref{fig:plotlanding}, while the UAV seems to reach the dotted line representing the correct diagonal position, it is actually slightly off and therefore violates the STL specification.
The minimal trace repair highlights this by adjusting the UAV's position closer to the dotted line while keeping the rest of the trace unchanged.

In theory, the repair problem could be solved using mixed-integer linear programming (MILP), where the system model, the trace, and the relevant parts of the TBT specification are encoded as constraints. To compute an ``optimal repair'' we impose a cost function, such as the $L_1$ norm, that penalizes deviation of the repaired trace states from the original trace states.   Since the size of the MILP grows with the length of the trace, the overall complexity of solving the MILP is exponential in the size of the trace. As a result, the MILP encoding approach does not scale.

We address this problem with two complementary repair strategies.
The first strategy,  \emph{incremental repair}, uses the segmentation provided by the TBT monitor to avoid encoding the entire trace.
Instead, it locally repairs segments, starting with leaf nodes and incrementally moving up the TBT only if necessary.
In our experiments, for a landing depicted in Figure \ref{fig:plotlanding} with a trace length of $1014$, this approach successfully repairs within $30$ seconds, whereas the straightforward MILP encoding takes over $2000$ seconds.
The second strategy, \emph{landmark-based repair}, is an iterative approach that identifies candidate landmarks within the trace for the repair.
If a repaired trace is found that satisfies the landmark, then its satisfaction guarantees the satisfaction of the TBT specification.
Given the abundance of candidates, we use the robust semantics of TBTs as a heuristic to find good candidates.
This approach entirely avoids the integer and binary variables introduced by the TBT encoding, enabling the problem to be solved as a linear program instead.
As a result, landmark-based repair successfully finds a solution for a trace with over $25,000$ entries in under ten minutes, where approaches that fully encode the TBT run out of memory.

The remainder of the paper is structured as follows: Section~\ref{sec:preliminaries} provides background on system models and TBTs, then Section~\ref{sec:trace_repair} introduces the repair strategies, finally Section~\ref{sec:experiments} demonstrates experimental results.

\section{Preliminaries}
\label{sec:preliminaries}
We revisit the definition of a system model and temporal behavior trees. 

\subsection{System Model}
We consider discrete-time systems with linear dynamics, represented by 
\begin{equation}
\label{eq:sys}
    X_{t+1} = A \cdot X_{t} + B \cdot U_{t},
\end{equation}
where $X_t \in \mathbb{X} \subseteq \mathbb{R}^n$ is the state at time step $t$ with $n$ variables,  $u_t \in \mathbb{U} \subseteq \mathbb{R}^m$ is the control input at time $t$. $A$ is an 
$n \times n$ while $B$ is a $n \times m$ matrix \cite{BEMPORAD20023}. 

\begin{example}\label{exp:dynamic_system}
Consider a one-dimensional integrator system with position $p$ and velocity $v$ as states, i.e., $X = [p,v]^T$, and an acceleration command as the input $U$. For a sampling time $t_s$, the system is described as
\begin{equation*}
    \begin{bmatrix}
    p_{t+1} \\ v_{t+1}
    \end{bmatrix} = \begin{bmatrix}
    1 & t_s \\ 0 & 1
    \end{bmatrix} \begin{bmatrix}
    p_{t} \\ v_{t}
    \end{bmatrix} + \begin{bmatrix}
    t_s^2/2 \\ t_s
    \end{bmatrix} U_t.
\end{equation*}
\end{example}
The execution of a system provides a \emph{trace} $\sigma$ that is a finite sequence of states $\sigma(1), \dots, \sigma(N)$, where $\sigma(i)$ represents the state of the model and the control input at timesteps $i \in \{1,\ldots,N\}$, i.e., $[X_{i-1},U_{i-1}]^T$. 
To access the state of the model $X_{i-1}$ at timestep $i$ we define a projection operation $\pi_X$ that extracts this state from $\sigma(i)$: $\pi_X(\sigma(i)) = X_{i-1}$.
The length of a trace is denoted as $|\sigma|$, with $|\sigma| = 0$ indicating an \emph{empty trace}.
Given a trace $\sigma$, the expression $\sigma[l:u]$ retrieves a slice of the trace from offsets $l$ to $u$ with $l,u \in \mathbb{N}_0$:
\[ 
\sigma[l:u] ::= \begin{cases}
\sigma(l+1), \cdots, \sigma(\min(u+1, N)), & \text{if}\ l \leq u\ \text{and}\ l < N \\ 
\text{empty trace} & \text{otherwise}\\
    \end{cases}
\]
For convenience, let $\sigma[l:] = \sigma[l: N-1]$ and $\sigma[:u] = \sigma[0:u]$. 

\begin{example}
Consider the system described in Example \ref{exp:dynamic_system}.
Let $t_s = 1$s, $X_0 = [0,1]^T$ and  $U_t = 0$ for all timesteps $t$. 
Executing the system for three timesteps yields $X_1 = [1,1]^T$, $X_2 = [2,1]^T$, and $X_3 = [3,1]^T$.
The corresponding trace has a length of four where $\sigma(1) = [0,1,0]^T$, $\sigma(2) = [1,1,0]^T$, $\sigma(3) = [2,1,0]^T$, and $\sigma(4) = [3,1,0]^T$. 
To exclude the first trace entry, slicing can be applied: $\sigma[1:] = \sigma(2), \sigma(3), \sigma(4)$.
Similarily, $\sigma[3:3] = \sigma(4)$ only accesses the last entry.
\end{example}

\subsection{Temporal Behavior Trees}
Temporal Behavior Trees (TBTs)~\cite{10.1145/3641513.3650180} extend behavior trees~\cite{6606326,10.1162/ARTL_a_00192,DBLP:conf/sle/GhzouliBJDW20,10155191,9921558,9355019}, as commonly used in robotics, with temporal formulas that are added to the leaf nodes of the trees. 

\begin{definition}
    [Syntax of Temporal Behavior Trees \cite{10.1145/3641513.3650180}]
    We construct a temporal behavior tree $\tree$ using the following syntax:
     \[\def\arraystretch{1.2} \begin{array}{rll}
    \tree & ::= \Fallback([\tree, \dots, \tree]), & \leftarrow \text{Fallback node} \\
    & |\ \Parallel_M( [\tree, \dots, \tree]), M \in \mathbb{N} & \leftarrow \text{Parallel node}\\
    & |\ \Sequence( [\tree, \tree]), & \leftarrow \text{Sequence node} \\
    & |\ \Leaf(\varphi), & \leftarrow \text{Leaf node} \\ 
    \end{array}\]
    where $\varphi$ is a local property expressed in a temporal language.
    Note that we restrict our focus to a selected subset of TBT operators for brevity and clarity. 
\end{definition}

Informally, $\Fallback([\tree_1, \ldots, \tree_n])$ mimics the semantics of a ``fallback'' node in a behavior tree: at least one of the subtrees $\tree_1, \ldots, \tree_n$ must eventually be satisfied by the trace. 
$\Parallel_M([\tree_1, \dots, \tree_n])$ denotes a parallel operator that specifies that at least $M$ distinct subtrees must be satisfied simultaneously by the trace $\sigma$. 
$\Sequence([\tree_1, \tree_2])$ is a sequential node that denotes that $\sigma$ must be partitioned into two  parts $\sigma_1 ; \sigma_2$ such that $\sigma_1 \models \tree_1$ and $\sigma_2 \models \tree_2$. 
For convenience, we rewrite $\Sequence([\tree_1, (\Sequence([\tree_2, \cdots, \Sequence( [\tree_{n-1}, \tree_n]))))$ for $n \geq 2$ as $\Sequence([\tree_1, \ldots, \tree_n])$.
Finally, a local formula at a leaf node specifies that the trace must satisfy the formula. 
For brevity, we occasionally omit $\Leaf(\varphi)$ and simply write $\varphi$.

We use Signal Temporal Logic (STL) to express local properties. 
STL is a prominent specification language for cyber-physical systems.
STL formulas are defined recursively as follows:
\[
    \varphi ::= \ap\ |\ \neg \varphi\ |\ \varphi \land \varphi\ |\ \varphi \lor \varphi\ |\ \F_{[l,u]} ( \varphi )\ |\ \G_{[l,u]} (\varphi)\ |\ \varphi \ \Un_{[l,u]}\ \varphi
\]
with an interval $[l,u]$ wherein $0 \leq l \leq u \leq \infty$.
Note that $\X(\varphi)$ is a shorthand for $\F_{[1,1]}$ as is $\F (\varphi)$ for $\F_{[0, \infty)}$.
Similar conventions apply for $\G (\varphi)$ and $\varphi_1 \Un \varphi_2$.
The set of \emph{atomic propositions} is $\ap \in \{p_1, \ldots, p_m\}$, with each $p_i$ associated with a function $f_i$ that maps states to a real number.
In this work, $f_i$ is restricted to be a linear function
A trace $\sigma$ satisfies $p_i$ if and only if $f_i(\sigma(1)) \geq 0$. 
Informally, STL extends propositional logic by temporal operators, where $\F_{[l,u]} ( \varphi )$ requires $\varphi$ to be satisfied eventually within $l$ and $u$ steps, $\G_{[l,u]} ( \varphi )$ requires $\varphi$ to be satisfied always within $l$ and $u$ steps, and $\varphi_1 \ \Un_{[l,u]}\ \varphi_2$ requires $\varphi_1$ to hold \emph{until}  $\varphi_2$ is eventually satisfied within $l$ to $u$ steps. 
Next, we formally define the satisfaction of a TBT specification that uses STL as specification language for its leaf nodes.

\begin{definition}
    [Boolean Semantics of Temporal Behavior Trees \cite{10.1145/3641513.3650180}]
\label{def:satTBT}
    The satisfaction of a TBT specification $\tree$ by a trace $\sigma$, denoted $\sigma \models \tree$, is defined as follows:
\[
\def\arraystretch{1.2} \begin{array}{lcl}
 \sigma \models \Fallback([\tree_1, \dots, \tree_n]) & \iff & \exists j \in [1,  n], \exists i \in [0,  |\sigma|-1],  \sigma[i:] \models \tree_j\\
 \sigma \models \Parallel_M( [\tree_1, \dots, \tree_n]) & \iff &  \exists I \subseteq [1,  n], |I| \geq M \land \forall i \in I, \sigma \models \tree_{i}\\
 \sigma \models \Sequence( [\tree_1, \tree_2]) & \iff &  \exists i \in [0,  |\sigma|-1], \sigma[: i] \models T_1 \land \sigma[i+1:] \models T_2\\
 \sigma \models \Leaf(\varphi), & \iff & \sigma \models \varphi  \\ 
 \sigma \models p_i\ & \iff & f_i(\sigma(1)) \geq 0 ~\mbox{if}~ |\sigma| > 0 ~\mbox{else}~ \false\\
 \sigma \models \neg \varphi & \iff &  \sigma \not\models \varphi\\
 \sigma \models \varphi_1 \land \varphi_2 & \iff & \sigma \models \varphi_1 \land \sigma \models \varphi_2 \\
 \sigma \models \varphi_1 \lor \varphi_2 & \iff & \sigma \models \varphi_1 \lor \sigma \models \varphi_2 \\
 \sigma \models \F_{[l,u]} ( \varphi ) & \iff & \exists\ i \in [l, u],\ \sigma[i:] \models \varphi \\
 \sigma \models \G_{[l,u]} (\varphi) & \iff & \forall\ i \in [l,u],\ \sigma[i:] \models \varphi \\
 \sigma \models \varphi_1 \ \Un_{[l,u]}\ \varphi_2 & \iff & \exists i \in [l,u],\ (\forall j \in [0, i-1],\ \sigma[j:] \models\ \varphi_1)\ \\ 
 & & \land\ \sigma[i:] \models \varphi_2 \\
 \end{array} 
\]
\end{definition}

Note that by Definition \ref{def:satTBT}, the formula $\neg\X true$ is true only  in the last state of a finite trace.
We abbreviate this formula by $\bullet$.

\newpage
\begin{example}\label{ex:oranges}
Consider a robot tasked with searching for objects, specifically apples and oranges\footnote{\url{https://robohub.org/introduction-to-behavior-trees/}}, while avoiding certain areas.
Assume there are three known locations, $L_1, L_2$, and $L_3$, where these objects are most likely to be found. The robot needs to search these locations in some order and discover either 
an apple or an orange in each location within a specified time limit. 
We can encode this search task using a TBT, as shown in Figure \ref{fig:exampletbt}.\vspace{-0.3mm}
The $\ap$s $\found$\drawApple{-0.2} and $\found$\drawOrange{-0.2} return $1$ if the robots reports finding an apple or an orange, respectively; otherwise, they return $-1$. 
The $\ap$s $\atLfirst$, $\atLsecond$, and $\atLthird$ denote that the robot is at positions $\mathit{L_1}$, $\mathit{L_2}$, and $\mathit{L_3}$, respectively. 
Similar, the $\ap$  $\avoidArea$ is satisfied iff the robot is outside the areas.
Note that according to the TBT, the order in which the robot finds an apple, an orange, or both is not relevant for successfully completing the task.
\end{example}

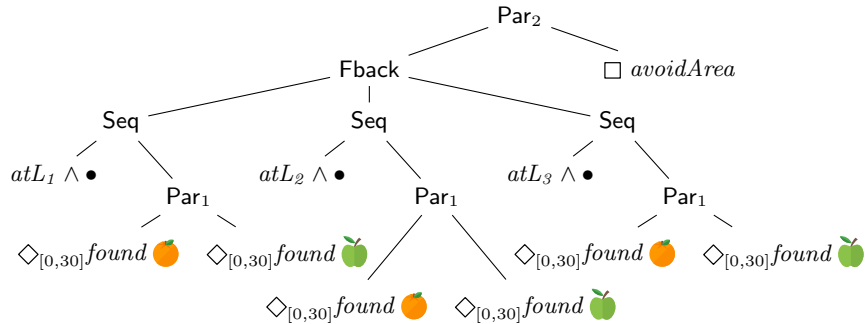
\begin{figure}
    \centering
    \begin{tikzpicture}
        \node {$\Parallel_2$}[sibling distance=40mm, level distance=7mm]
            child { node {$\Fallback$}[sibling distance=33mm, level distance=7mm]
                child {node {$\Sequence$}[sibling distance=18mm, level distance=7mm]
                child {node {$\atLfirst \wedge \bullet$}}
                child[level distance=10mm] {node {$\Parallel_1$}[sibling distance=25mm, level distance=8mm]
                    child {node { \drawOrangeWithText{-1.1}{-0.04}{$\F_{[0,30]} \found$}}}
                    child {node { \drawAppleWithText{-1.1}{-0.04}{$\F_{[0,30]} \found$}}}                  
                    }
                }
                child {node {$\Sequence$}[sibling distance=18mm, level distance=7mm]
                child {node {$\atLsecond \wedge \bullet$}}
                child[level distance=10mm] {node {$\Parallel_1$}[sibling distance=25mm, level distance=15mm]
                    child {node { \drawOrangeWithText{-1.1}{-0.04}{$\F_{[0,30]} \found$}}}
                    child {node { \drawAppleWithText{-1.1}{-0.04}{$\F_{[0,30]} \found$}}}
                    }
                }
                child {node {$\Sequence$}[sibling distance=18mm, level distance=7mm]
                child {node {$\atLthird \wedge \bullet$}}
                child[level distance=10mm] {node {$\Parallel_1$}[sibling distance=25mm, level distance=8mm]
                    child {node { \drawOrangeWithText{-1.1}{-0.04}{$\F_{[0,30]} \found$}}}
                    child {node { \drawAppleWithText{-1.1}{-0.04}{$\F_{[0,30]} \found$}}}
                    }
                }
            }
            child {node {$\G\ \avoidArea$}}
        ;
    \end{tikzpicture}
    \caption{TBT to search for at least an apple or an orange at different locations.}
    \label{fig:exampletbt}
\end{figure}
\section{Repairing Traces}
\label{sec:trace_repair}
Given a system model $M$, a TBT specification $\tree$, and a trace $\sigma$ with $\sigma \not \models \tree$, we construct a repaired trace $\sigma_R$ that satisfies $\tree$ with $|\sigma_R| = |\sigma|$, denoted by $\Repair(\tree, M, \sigma$, $\xi)$, as follows:
\begin{equation}
\label{eq:repair}
\left.\begin{array}{rll}\arg\min_{\sigma_{R}} & J(\sigma, \sigma_R) & \leftarrow \text{Minimize cost function}\\  
     s.t. & \sigma_R \models \Encode_{\sigma_R}(M)& \leftarrow \text{Trace follows system model }\\
& \sigma_R \models \Encode_{\sigma_R}(\tree)& \leftarrow \text{Trace satisfies TBT}\\
     & \sigma_R \models \xi & \leftarrow \text{Additional constraints, cf.\ Section \ref{sec:repairs}}\\
\end{array} \right\}
\end{equation}
Note that $\Repair$ does not change the length of the trace: each state of the repaired trace is one to one correspondent to a state of the original trace.
A repair problem can be infeasible. For instance,  one may need a strictly longer trace $\sigma_R$ to satisfy $\tree$. In this case, $\Repair$ returns \textsf{none}.
In the following, we introduce $\Encode_{\sigma_R}(M)$, $\Encode_{\sigma_R}(\tree)$, and present different cost functions $J$.
We then present complementary repair strategies.
We assume $\sigma$ and $\sigma_R$ are encoded using $2 \cdot |\sigma|$ continuous variables with $|\sigma|$ constraints for the original trace.

\paragraph{Encoding of the System Model:}
Since we consider discrete-time systems with linear dynamics, Equation \eqref{eq:sys} can be directly encoded as a MILP with $X_t$ being the states of the repaired trace $\sigma_R$.
In this way we ensure that $\sigma_R$ follows the dynamics.
We furthermore introduce lower and upper bounds for each control input in $U_t$.
We denote this encoding by $\Encode_{\sigma_R}(M)$.
The encoding $\Encode_{\sigma_R}(M)$ adds $3 \cdot |\sigma_R|$ constraints.
These constraints include the system model itself, a lower bound on the inputs, and an upper bound on the inputs.
Since $A$ and $B$ are constants, the encoding adds $|\sigma|$ continuous variables for $U$.

\paragraph{Encoding of the Temporal Behavior Tree:}
Since the length of the trace is known, the satisfaction of a trace with respect to a TBT specification, as defined in Definition \ref{def:satTBT}, can be encoded in a MILP by extending the MILP encoding for STL to TBT.
For a detailed explanation for STL, we refer to \cite{raman2014model}.
Next, we provide a summary of the STL encoding before extending it to TBTs.

For each STL (sub)formula $\psi$, and for a trace $\sigma_R$ with $|\sigma_R|=N$, we introduce a binary variable $z^\psi_t$ for each trace position $0\leq t < N$.
The variable $z^\psi_t=1$ if and only if $\psi$ is satisfied at position $t$ in the trace.
Since the length of the trace is known, existential quantifiers $\exists j \in [1,\dots,N]$ and universal quantifiers $\forall j \in [1,\dots,N]$ are encoded by $\bigvee_{j=1}^N$ and  $\bigwedge_{j=1}^N$, respectively.
Thus, the remaining task is to encode $\vee$ and $\wedge$, which can be done as in \cite{raman2014model}: 
\[
    \psi = \bigwedge_{i=1}^{m} \varphi_i: 
    \begin{cases}
        z_{t}^\psi \leq z_{t}^{\varphi_i}, i \in [1,\dots,m] &\\
        z_{t}^\psi \geq 1 - m + \sum_{i=1}^{m}  z_{t}^{\varphi_i} \\
    \end{cases}
    \psi = \bigvee_{i=1}^{m} \varphi_i: 
    \begin{cases}
        z_{t}^\psi \geq z_{t}^{\varphi_i}, i \in [1,\dots,m] &\\
        z_{t}^\psi \leq \sum_{i=1}^{m}  z_{t}^{\varphi_i} \\
    \end{cases}
\]
To extend the encoding to TBTs $\tree$, we introduce binary variables that account for the current segment: $z^\tree_{[t_1:t_2]}$. 
The encoding is as follows:
\begin{equation}
\label{eq:tbtencoding}
\begin{array}{ll}
    \tree = \Leaf(\varphi)
        & : z_{[t_1:t_2]}^{\tree} = z_{[t_1:t_2]}^{\varphi}\\
    \tree = \Fallback([\tree_1, \dots, \tree_n])
        & : z_{[t_1:t_2]}^{\tree} = \bigvee_{j=1}^{n}\bigvee_{i=t_1}^{t_2} z_{[i:t_2]}^{\tree_j}  \\
    \tree = \Parallel_M( [\tree_1, \dots, \tree_n])
        & : z_{[t_1:t_2]}^{\tree} = (\sum_{j=1}^n z_{[t_1:t_2]}^{\tree_{j}}) \geq M\\
    \tree = \Sequence( [\tree_1, \tree_2])
        & : z_{[t_1:t_2]}^{\tree} = \bigvee_{i=t_1}^{t_2-1} (z_{[t_1:i]}^{\tree_1} \wedge z_{[i+1:t_2]}^{\tree_2})\\
\end{array}
\end{equation}

We denote this encoding by $\Encode_{\sigma}(\tree)$.
A trace $\sigma$ satisfies a specification $\tree$ iff $z_{[0:N-1]}^\tree = 1$.
Let $|\varphi|$ represent the size of the temporal formula. 
Unlike the previous STL encoding, which introduces $(|\ap| + |\varphi|) \cdot N$ binary variables, this encoding requires $(|\ap| + |\tree|) \cdot N^2$ binary variable to account for the different segments with the same number of constraints that track their satisfaction.

\newpage
\begin{example}
Consider the TBT $\Sequence([\Leaf(\F\G \hasKey), \Leaf(\openDoor)])$, abbreviated by $\Sequence(\dots)$, that specifies that a robot must first find a key before opening the door.
Given a trace of length ten, to check whether the trace satisfied the TBT, i.e., $z^{\Sequence(\dots)}_{[0:9]} = 1$, it is necessary to search for a satisfying transition between its children.
We therefore compute $\bigvee_{i=0}^{8} (z_{[0:i]}^{\Leaf(\F\G \hasKey)} \wedge z_{[i+1:9]}^{\Leaf(\openDoor)})$.
\end{example}

\paragraph{Encoding of Cost Functions}
The cost function in Equation \eqref{eq:repair} ensures that the repaired trace $\sigma_R$ is \emph{similar} to the original trace $\sigma$, thereby providing a clear and intuitive explanation of \emph{what should have been done differently} to satisfy the TBT $\tree$ specification.
We consider the following cost functions:
\begin{itemize}
    \item \emph{L1-Distance} $L1(\sigma, \sigma_R) = \sum_{i=1}^{|\sigma|} \|\pi_X(\sigma(i)) - \pi_X(\sigma_R(i))\|_1$\\
    This metric computes the point-wise distance between both trajectories. 
    \item \emph{Hamming-Distance} $H(\sigma, \sigma_R) = \sum_{i=1}^{|\sigma|} \begin{cases}
        1 &, \pi_X(\sigma(i)) \neq \pi_X(\sigma_R(i))\\
        0 &, \text{otherwise}.
    \end{cases}$\\
    This metrics counts the changes necessary to the original trajectory.
\end{itemize}
We also consider the robust semantics of TBTs as cost function.
Instead of providing a Boolean verdict, the robust semantics of a TBT yields a numerical value.
A positive value indicates that the specification is satisfied, while a negative value corresponds to a violation.
Additionally, the magnitude of the numerical value reflects the degree to which the specification is satisfied or violated.
For a formal introduction of robust semantics of TBTs, we refer to \cite{10.1145/3641513.3650180}.
In essence, the robust semantics of TBTs can be viewed as a variation of the Boolean semantics described in Definition \ref{def:satTBT}.
Specifically, in the robust semantics, all instances of $\exists$ and $\vee$ are replaced by $\max$, while all instances of $\forall$ and $\wedge$ are replaced by $\min$.
This transformation allows the semantics to yield a numerical value that indicates the degree of satisfaction that can be used as cost function $R(\sigma_R)$ as in \cite{raman2014model} to maximize ``satisfaction''. 
We also introduce a weighted combination $W(\sigma, \sigma_R)$, which integrates multiple objectives and constraints into the repair strategy: 
$\tau_{L1} \cdot L1(\sigma, \sigma_r) + \tau_{H} \cdot H(\sigma, \sigma_r) + \tau_{R} \cdot R(\sigma_r)$ where $\tau_{L1}$, $\tau_{H}$, and $\tau_{R}$ are positive weights that sum up to one.
Note that we can encode $|x|$ using linear constraints within the cost function by defining an auxiliary variable $a = |x|:  a \geq x \wedge a \geq -x$ or using the Big-M method~\cite{williamsModelBuildingMathematical2013}.
For brevity, the paper focuses on cost functions that are intuitive for pilots and control engineers.
Other cost functions that find repairs using shorter traces will be explored in future.

\subsection{Repair Strategies} \label{sec:repairs}
We introduce two repair strategies, to address two distinct scalability concerns.
The first strategy is an incremental approach, enabling local repairs of violating trace segments, which avoids  encoding the full trace.
The second strategy utilizes landmarks to resolve choices introduced by disjunctions in the specification.
This reduces the MILP encoding to a  linear program, which allows for more efficient optimization algorithms, such as the simplex algorithm \cite{vanderbeiLinearProgrammingFoundations2020}.
The strategies can be applied both in isolation and in combination with each other.

\subsubsection{Incremental Repair Strategy}

Incremental repair strategy is based on a \emph{segmentation} of the original violating trace with respect to the TBT. 
The idea of segmentation was developed in \cite{10.1145/3641513.3650180}, wherein a segmentation of a trace w.r.t.\ a TBT shows how the overall robustness of the trace w.r.t.\ a specification $\tree$ can be decomposed into the robustness of sub-traces with respect to sub-trees of $\tree$.  
Incremental repair uses the segmentation to attempt \emph{local repairs} of parts of the trace w.r.t.\ parts of the specification. If the local repair fails, the  incremental repair widens the scope of the repair iteratively, falling back on the original full repair of the trace w.r.t.\ the entire specification in the worst case. 

First, we recall the notion of a segmentation.
A segmentation of a trace with respect to a TBT $\tree$ divides the trace $\sigma$ into multiple subtraces of the form $\sigma[i: j]$ and assigns a (sub)tree to each of them.

\begin{definition}[Segmentation of a TBT \cite{10.1145/3641513.3650180}]\label{def:segmentation}
   The segmentation of a trace $\sigma$ with respect to a TBT $\tree$ is a directed acyclic graph $G = (V,E)$ whose vertex set  $V$ consists of  triples of the form 
   \[ V = \{ (\hat{\tree}, i, j)\ |\ \hat{\tree}\ \text{is a subtree of}\ \tree,\ 0 \leq i \leq j \leq |\sigma|-1 \} \,,\]
   and edges $E \subseteq V \times V$
such that the following conditions hold:
\begin{enumerate}
\item $(\tree, 0, |\sigma|-1) \in V$ corresponding to  the entire tree $\tree$ and the entire trace from indices $1$ to $|\sigma|$.
\item If a node $v$ is of the form  $(\Fallback([\tree_1, \ldots, \tree_k]), i, j) \in V$, and $i \leq j$  then there is precisely one subtree index $l \in [1, k]$ and a single trace index $i' \in [i, j]$ such that the edge $v \rightarrow (\tree_l, i', j) \in E$. 
\item If a node $v$ of the form  $(\Sequence([\tree_1, \tree_2]), i, j) \in V$, there exists a unique index $u$ such that $(\tree_1, i, u) \in V$, $(\tree_2, u+1, j) \in V$ and the edges $v \rightarrow (\tree_1, i, u) $ and $v \rightarrow (\tree_2, u+1, j) $ belong to $E$. 
\item If a node $v$ of the form $(\Parallel_M([\tree_1, \ldots, \tree_k]), i, j) \in V$, we have $M$ distinct indices $l_1,\ldots, l_M \in [1, k]$ such that the set $S = \{ (\tree_{l_1}, i, j), \cdots, (\tree_{l_M}, i, j)\} \subseteq V$ and edges from $v$ to each of the nodes in $S$ belong to $E$.
\item The set of vertices $V$ and edges $E$ are minimal: i.e, no proper subsets of $V, E$ satisfy the conditions stated above.
\end{enumerate}
\end{definition}
Notice the correspondence between the nodes in a segmentation and the notion of a trace satisfying/violating a TBT from Definition~\ref{def:satTBT}.
\begin{example}\label{ex:segmentation}
Figure \ref{fig:segmentation} provides a segmentation for a trace of length 100 using the TBT specification depicted in Figure \ref{fig:exampletbt}.
The segmentation shows that the robot tries to reach location $L_2$ and then tries to find an orange.
To check whether the execution was successful, the satisfaction of each leaf node must be evaluated. 
\end{example}

\begin{figure}
    \centering
    \begin{tikzpicture}
        \node {$(\Parallel_2, 0, 99)$}[sibling distance=40mm, level distance=7mm]
            child { node {$(\Fallback, 0, 99)$}[sibling distance=35mm, level distance=7mm]
                child {node {$(\Sequence, 60, 99)$}[sibling distance=50mm, level distance=7mm]
                child {node {$(\Leaf(\atLsecond \wedge \bullet), 60, 60)$}}
                child {node {$(\Parallel_1, 61, 99)$}[sibling distance=35mm, level distance=10mm]
                    child {node { \drawOrangeWithText{-0.7}{-0.04}{$(\Leaf(\F_{[0,30]} \found~~~~),61, 99)$}}}
                    }
                }
            }
            child {node {$(\Leaf(\G\ \avoidArea), 0, 99)$}}
        ;
    \end{tikzpicture}
    \caption{Segmentation graph for the TBT in Figure \ref{fig:exampletbt}.}
    \label{fig:segmentation}
\end{figure}
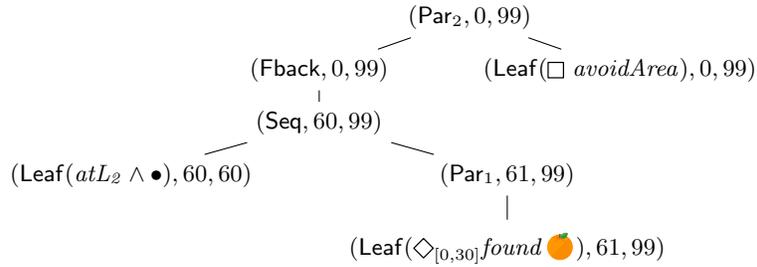

From the construction of a TBT, we establish an important result that satisfaction of a TBT reduces to checking the temporal formulas at the leaves given the segmentation.
\begin{proposition}
$\sigma \models_G \tree$ if and only if for all $(\Leaf(\varphi), i, j) \in V, \sigma[i:j] \models \varphi$. 
\end{proposition}
The proof of the proposition follows by induction on the structure of the TBT and matching the semantics of TBT (Definition~\ref{def:satTBT}) against the definition of a segmentation (Definition~\ref{def:segmentation}). However, an ``optimal'' segmentation can be defined and computed even for violating traces, i.e., a subtrace violates its assigned subtree. Such a segmentation provides useful clues for repairing the trace.
\begin{example}
\label{ex:segmentationorange}
We now evaluate Example \ref{ex:segmentation}, assuming all leaves satisfy their STL formula with respect to their respective segment, except $\Leaf(\G\ \avoidArea)$.\vspace{-0.8mm}
Since $\vspace{-1mm}((\Parallel_1, 61, 99), \raisebox{-0.4\height}{\begin{tikzpicture}
        \path pic {orange};
       \node at (-0.8, -0.04) {$(\Leaf(\F_{[0,30]} \found~~~~),61, 99)$}; 
    \end{tikzpicture}}) \in E$ and the leaf is satisfied, it follows that $(\Parallel_1, 61, 99)$ is satisfied.
The same holds for $(\Sequence, 60, 99)$ and $(\Fallback, 0, 99)$: both are satisfied.
However, since $(\Leaf(\G\ \avoidArea), 0, 99)$ is not satisfied, $(\Parallel_2, 0, 99)$ is also not satisfied. 
Ideally, we only need to repair $\Leaf(\G\ \avoidArea)$ to satisfy the TBT specification. 
\end{example}

Next, we present the incremental repair strategy.
This strategy first tries to repair violating segments \emph{locally} with respect to their respective leaf nodes. In doing so, we may need to ensure that there is a valid transition from the end states of the repaired traces to the states of the original trace, so that the repaired segment can be substituted into the original trace.
If the local repair is unsuccessful, it incrementally moves up to the parent node according to the segmentation graph, attempting to fix segment transitions first, followed by adjusting segment boundaries only if necessary.
We define the local repair strategy as follows.

\begin{definition}[Local Repair]
\label{def:local}
Let $M$ be a system model, $\sigma$ be a trace, $G=(V,E)$ be a segmentation graph with $V' \subseteq V$, $f \in \{\valid, \loose\}$ indicate whether a valid transition to the next segment is required, and let $\textsf{validTransition}_M(\sigma, \sigma_R)$ be a function that checks for a valid transition from the repaired trace segment to the original trace.
We repair a trace locally, denoted $\Local(M, \sigma, V', f)$, as follows: 
\[
    \bigoplus_{(\tree,i,j)  \in V'} \Repair(\tree, M, \sigma[i:j], \textsf{validTransition}_M(\sigma, \sigma_R) \vee f = \loose) 
\]
where $\bigoplus$ merges the models, enforcing all constraints while minimizing the sum of the individual costs. 
Note that $\Repair$ is defined in Equation \eqref{eq:repair}.
\end{definition}

If no leaf node segments overlap, an efficient approach is to locally repair each leaf node and to update the trace when initial or last segment states change.
However, for more complex cases like the segmentation graph in Figure \ref{fig:segmentation} this is insufficient, as locally repairing $(\Leaf(\G\ avoidArea), 0, 99)$ can impact the satisfaction of the other leaf nodes.
The next repair strategy, which we call $\Incremental$ repair, accounts for that.
The $\Incremental$ repair strategy,  shown below as Algorithm \ref{alg:incremental}, repairs a given segmentation $G$.
It initializes two sets: $C$ for the repaired segments (Line 1) and $L$ for the segments to be repaired (Line 2).
Initially, $L$ contains all leaves in $G$ that violate its segment.
Since segmentation may have overlapping violating segments that must be repaired together to avoid side effects, e.g., $(\Leaf(\G\ \avoidArea), 0, 99)$ in Figure \ref{fig:segmentation}, $\textsf{mergeOverlap}_G(L)$ merges these overlapping sets in $L$ and removes entries if an ancestor node is already in the set.
Afterwards one set is removed from $L$ (Line 5) to be locally repaired next (Line 6 and 8).
If there is a existing repair (Line 10), then $\textsf{affectedLeaves}(\sigma, \sigma_R, V)$ checks if other leaves are affected by changes of $\sigma_R$.
This also includes the check of valid transition between segments in case of a \emph{loose} local repair (Line 8).
For instance consider Example \ref{ex:segmentationorange}, in the beginning the set $L$ is $\{ \{(\Leaf(\G\ \avoidArea), 0, 99)\}\}$ and if $\sigma_R[61:99] \neq \sigma[61:99]$ then the  evaluation of
$\vspace{-1.4mm}\raisebox{-0.4\height}{\begin{tikzpicture}
        \path pic {orange};
       \node at (-0.8, -0.04) {$(\Leaf(\F_{[0,30]} \found~~~~),61, 99)$};
    \end{tikzpicture}}$ might change. 
A repair is considered successful if only the leaves in $l$ are affected (Line 13).
Otherwise, a subsequent local repair must account for the affected leaves in the next iteration (Line 15).
If the local repairs were not successful (Line 18), then the repair moves to its next common ancestor w.r.t.\ the elements in $l$.
Finally, if there are no open repairs, i.e., $L=\emptyset$, we can read off $\sigma_R$ using the set $C$ (Line 22).

\begin{algorithm}[t]
\caption{The $\Incremental_G(M, \sigma)$ repair strategy.}
\label{alg:incremental}
\begin{algorithmic}[1]
\Require Segmentation $G=(V,E)$, a system model $M$, and trace $\sigma$
\Ensure Repaired trace $\sigma_R$ or $\textsf{none}$
\State $C \gets \{\}$ \Comment{Successful repairs}
\State $L \gets \{\{ v | v \in V \wedge v = (\Leaf(\varphi),i,j)  \wedge \sigma[i:j] \not\models \varphi  \}\}$ \Comment{Start with violating leaves} 
\State $L \gets \textsf{mergeOverlap}_G(L)$ \Comment{Groups overlapping leaves into one set} 
\While {$L \neq \emptyset$}
    \State $l \gets L.pop()$ \Comment{Set of nodes in $V$ that must be repaired }
    \State $\sigma_R \gets \Local(M, \sigma, l, valid)$ \Comment{If successful, it avoids affecting other leaves}
    \If{$\sigma_R = \textsf{none}$}
        \State $\sigma_R \gets \Local(M, \sigma, l, loose)$ \Comment{Could affect transitions to other leaves}
    \EndIf
    \If{$\sigma_R \neq \textsf{none}$} \Comment{There is a repair}
        \State $\text{affected} \gets \textsf{affectedLeaves}(\sigma, \sigma_R, V)$ \Comment{Tracks leaf changes due to $\sigma_R \neq \sigma$}
        \If{$\text{affected} \setminus l = \emptyset$}
            \State $C \gets C \cup (l, \sigma_R$) \Comment{Successful repair!}
        \Else
            \State $L = L \cup \text{affected}$ \Comment{Need to account for affected leaves}
        \EndIf
    \Else  \Comment{There is no repair, therefore we need to move up the TBT}
        \State $L \gets L \cup \textsf{getCommonAncestor}_G(l)$
    \EndIf
    \State $L \gets \textsf{mergeOverlap}_G(L)$
\EndWhile
\State \Return $\textsf{compose}(C, \sigma)$
\end{algorithmic} 
\end{algorithm}

The $\Incremental$ repair leverages information from the segmentation graph to avoid a costly exploration within the optimization model.
Note, however, that a segmentation omits the different choices for a $\Fallback$ node, which also means that the $\Incremental$ repair does not take these into account.
Fortunately, the approach presented in \cite{10.1145/3641513.3650180}, which utilizes dynamic programming, provides ``alternative segmentations'' for every choice of a $\Fallback$ node by simply ``reading off'' the table entries.
We denote the set of segmentations that contains all choices of $\Fallback$ nodes in a TBT $\tree$ as $\hat{G}$.

\begin{proposition}
If for all $G\in \hat{G}$, $\Incremental_G(M, \sigma) = \textsf{none}$ then there is no $\sigma_R$ with $|\sigma_R|=|\sigma|$ for which $\sigma_R \models \tree$.
\end{proposition}

\begin{theorem}
The $\Incremental$ repair strategy is sound w.r.t.\ a segmentation and terminates.
\end{theorem}
\begin{proof}
It is sound because if there are no successful local repairs, $\Incremental$ repair moves up its segmentation graph and finally repairs the root node (see Line 18 in Algorithm \ref{alg:incremental}). 
I.e., in the worst case, the MILP finds a trace $\sigma_R$ that satisfies the constraints on the TBT and the model, while ignoring the relation to the original trace $\sigma$. 
It only returns $\textsf{none}$ if there is no such trace $\sigma_R$ with $|\sigma_R|=|\sigma|$.
It terminates because $\textsf{mergeOverlap}$ avoids having segments that affect each other over and over again.
Assume $L = \{ V_1, V_2\}$.
There are three cases for an attempt to repair $V_2$: the repair was successful while no other node is affected (Line 13), the repair of $V_2$ was successful but there was an affection (Line 15), or there was no successful repair and $\Incremental$ moves up the tree (Line 18).
For the first case, the size of $L$ is reduced since one element was removed from $L$ (Line 5).
For the second case, there must be a node $v$ that is affected. 
Then, it either has an overlap with $V_1$ and according to Line 20 both sets are merged or there is no overlap and both sets $V_1$ and $V_2$ remain in $L$, while $V_2$ is increased by an additional element -- converging to repairing the whole segmentation.
The same holds for the last case.
\qed
\end{proof}

\paragraph{Landmark-Based Repair Strategy:} The incremental repair strategy can be complemented by an approach that deals with the disjunctive constraints encountered in the encoding of the repair problem~\eqref{eq:tbtencoding}. These constraints arise, for instance, when a fallback operator is encoded, or at a leaf node with a $\F$ formula. Given a disjunctive formula $\bigvee_{j=1}^m \psi_j$, the landmark based strategy uses information from the original trace to select a candidate $\psi_j$ that will be satisfied.

Formally, a candidate for a landmark is a minimal set of propositions that is sufficient to satisfy the TBT specification. 
This simplifies the optimization problem to a linear program for the $L_1$ and Hamming distance, as it no longer requires integer or binary variables, and only linear constraints remain\footnote{Using the robust semantics of TBTs requires the Big-M encoding of the absolute value function, which in turn adds binary variables.}. 
Outside the optimization problem, we can also use the robust semantics of TBTs to rank the candidates, i.e., we interpret robustness of \ap s as a heuristic.
However, in Section \ref{sec:experiments}, we will see that the repair provides a fast solution for the repair and improves upon that, similar to an anytime algorithm.

\begin{example}\label{ex:landmark}
Given a leaf node $\F_{[0,10]} \atLfirst$ with $f_{\atLfirst}(x) = x - 2.5$, $\Encode_{\sigma_R}(\tree)$ unfolds into a disjunction of $\atLfirst$ at the next positions.
The trace shown in Figure \ref{fig:candidates} violates the property because there is no value of $x$ at any position $i$ in the trace where $f_{\atLfirst}(x)$ yields a positive value.
Landmark-based repair will pick a candidate to solve the disjunction.
Here, eleven candidates for landmarks exist.
The best candidate w.r.t.\ the robustness value of $\atLfirst$ is at position $i=7$ with value close to $-0.5$ (highlighted by the green circle).
The landmark is encoded by $\sigma_R(8) \models \atLfirst$.
\end{example}

Algorithm \ref{alg:landmark} provides an iterative repair strategy based on landmarks.
The repair receives a segmentation, a system model, and the violating trace.
In Line~3, the set of candidates is computed and ordered.
It then takes the most promising candidate (Line~5), sets an upper bound on the repair cost (Line~6), encodes the repair (Line~7), and checks if the current landmark returns a better repair (Line~8).
Note that the constraints on the system model and cost function remain unchanged throughout the iteration (Line~4).
Therefore, the LP can efficiently updated by simply removing the previous landmark and adding the next one, avoiding the need to rebuild the LP from scratch.
For the experiments, two optimizations were implemented but are omitted here for brevity.
First, candidates are computed on-the-fly instead of upfront as robustness maximizes.
Second, candidates are ranked by robustness and explored while maintaining a minimum distance w.r.t.\ the position of previously tested candidates.
This distance starts large and gradually decreases whenever no candidates are left within this distance.
Once the distance reaches $1$ and no further candidates are available, the repair terminates.
For instance, considering Example \ref{ex:landmark} where position $i=7$ is chosen as first landmark, a distance of two excludes Positions $5,6,8,9$ and allows to explore Position $3$ next.
So far, the correct distance and its rate of decrease are user-defined parameters.

\begin{algorithm}
\caption{The $\Landmark_G(M, \sigma)$ repair strategy.}
\label{alg:landmark}
\begin{algorithmic}[1]
\Require A segmentation $G$, a system model $M$, and trace $\sigma$
\Ensure Repaired trace $\sigma_R$ or $\textsf{none}$
\State $\sigma_R \gets \textsf{none}$
\State $J_{upper} \gets \infty$
\State $\candidates \gets \computeCandidates(G, \sigma)$ \Comment{List of landmarks ordered by heuristic}
\While {$\candidates \neq \emptyset$}
    \State $\landmarks \gets \candidates.pop()$
    \State $\mathit{improve} \gets J(\sigma, \sigma_R) < J_{upper}$ \Comment{Upper bound on the cost of the repair}
    \State $\sigma_R' \gets \Repair(\true, M, \sigma, \landmarks \wedge \mathit{improve})$ \Comment{$\landmarks$\ replaces TBT}
    \If{$\sigma_R' \neq \textsf{none}$} \Comment{If there is a repair, it must be better}
        \State $\sigma_R \gets \sigma_R'$
        \State $J_{upper} \gets J(\sigma, \sigma_R')$
    \EndIf
\EndWhile
\State \Return $\sigma_R$ 
\end{algorithmic} 
\end{algorithm}

\paragraph{Discussion}
Both strategies can be combined to efficiently solve complex repairs.
The incremental repair allows to divide the trace into smaller segments, each of which can be repaired locally.
When a local property contains multiple disjunctions, e.g., when using an unbounded $\F$, the optimization becomes more challenging.
In such cases, we can use the landmark-based repair. 
If the combination of both strategies fails, refining the TBT specification is a viable option.

The closer the violating trace is to satisfying the specification, the more effective segmentation and landmarks are as starting points for the repair.
For instance, if the TBT consists of a sequence of two nodes, with the segmentation assigning a very short segment to the first node -- potentially too short for a repair -- and the remainder of the trace to the second node, incremental repair will converge to the full encoding. 
In such cases, starting with the full encoding or even finding a new unrelated trace would be faster. 

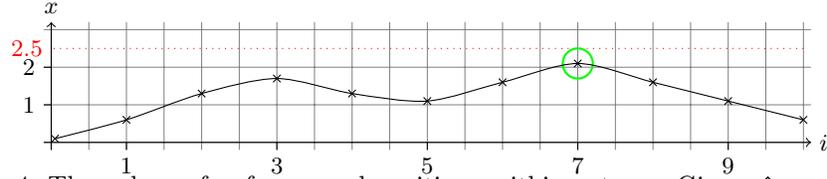
\begin{figure}[H]
\vspace{1mm}
    \centering
\begin{tikzpicture}
\draw[step=0.5cm, very thin,color=gray] (-0.1,-0.1) grid (10.1,1.6);
\draw[-,dotted,red] (10,1.25) -- (0,1.25) node[left]{2.5};
\draw[->] (0,-0.1) -- (0,1.6) node[above] {$x$};
\draw[->] (-0.1,0) -- (10.1,0) node[right] {$i$};
\draw[green, thick] (7,1.05) circle (0.2cm);
\foreach \x in {1, 3, 5, 7, 9} {
    \draw (\x, 0) -- (\x, -0.1) node[below] {\x};
}
\draw (0,0.5) -- (-0.1,0.5) node[left] {1};
\draw (0,1) -- (-0.1,1) node[left] {2};

\tikz \draw plot[mark=x,smooth] file {plots/ap.table};
\end{tikzpicture}
\caption{The values of $x$ for several positions within a trace. Given $\F_{[0,10]} \atLfirst$ with $f_{\atLfirst}(x) = x -2.5$, position $i=7$ is a good candidate for a landmark.}
\label{fig:candidates}
\end{figure}
\section{Empirical Evaluation and Case-Studies}
\label{sec:experiments}
This section presents two case studies: the robot search task shown in Figure \ref{fig:exampletbt} and the automated landing of a UAV on a ship.
The first case study illustrates the impact of different cost functions, while the second showcases incremental repair and compares it to landmark-based repair.
All experiments were run on a single 16-core machine with a 2.50 GHz $11^{th}$ Gen Intel(R) Core(TM) i7-11850H processor with 32 GB RAM.
The algorithms are implemented in Python using Gurobi\footnote{\url{https://www.gurobi.com/}: Gurobi Optimizer version 11.0.0 build v11.0.0rc2} as optimizer. 
Segmentations were obtained using the approach from~\cite{schirmer_2024_13807484}.

\subsection{Robot Search Task}
Using the robot search task introduced in Example \ref{fig:exampletbt}, we demonstrate two different cost functions using incremental repair.
The trace we used has a best segmentation in which both the location and one of the fruits were narrowly missed.
Also, the restricted area that was meant to be avoided was breached.
The original trace, consisting of $540$ entries, was reduced to $68$ entries through subsampling by \cite{schirmer_2024_13807484}, which computed the best segmentation in under a second.
Figure \ref{fig:differentcost} depicts the results of the incremental repair: on the left, using $\mathit{L1}$ as cost function, and on the right, using a weighted combination $W$ of $\mathit{L1}$ and robustness $R$ with weights $0.01$ and $0.99$, respectively.
It took $6$s to repair using $\mathit{L1}$ and $21$s using $W$.
Both repairs reach the location $L2$ and find a fruit (top left).
Note that the repair using $W$ provides larger separation from the restricted area compared to $\mathit{L1}$ but still resembles the original trace.
The runtime of the repair is mostly impacted by $\G\ \avoidArea$, as it relies on the whole trace.
Also, $\ap$ $\avoidArea$ adds constraints to keep points outside the region and maintain a minimum distance to the corners of the restricted areas.
This can be avoided by a syntactic reformulation of the TBT, optimizing it to provide segmentation graphs that are easier to repair. 
For instance, by moving the $\G\ \avoidArea$ invariant into the individual leaf nodes.
Yet, this is not the scope of this paper.

\begin{figure}[t]
\begin{minipage}{0.5\textwidth}
    \includegraphics[trim={4cm 2.0cm 4.5cm 3.0cm},clip,width=1.0\textwidth]{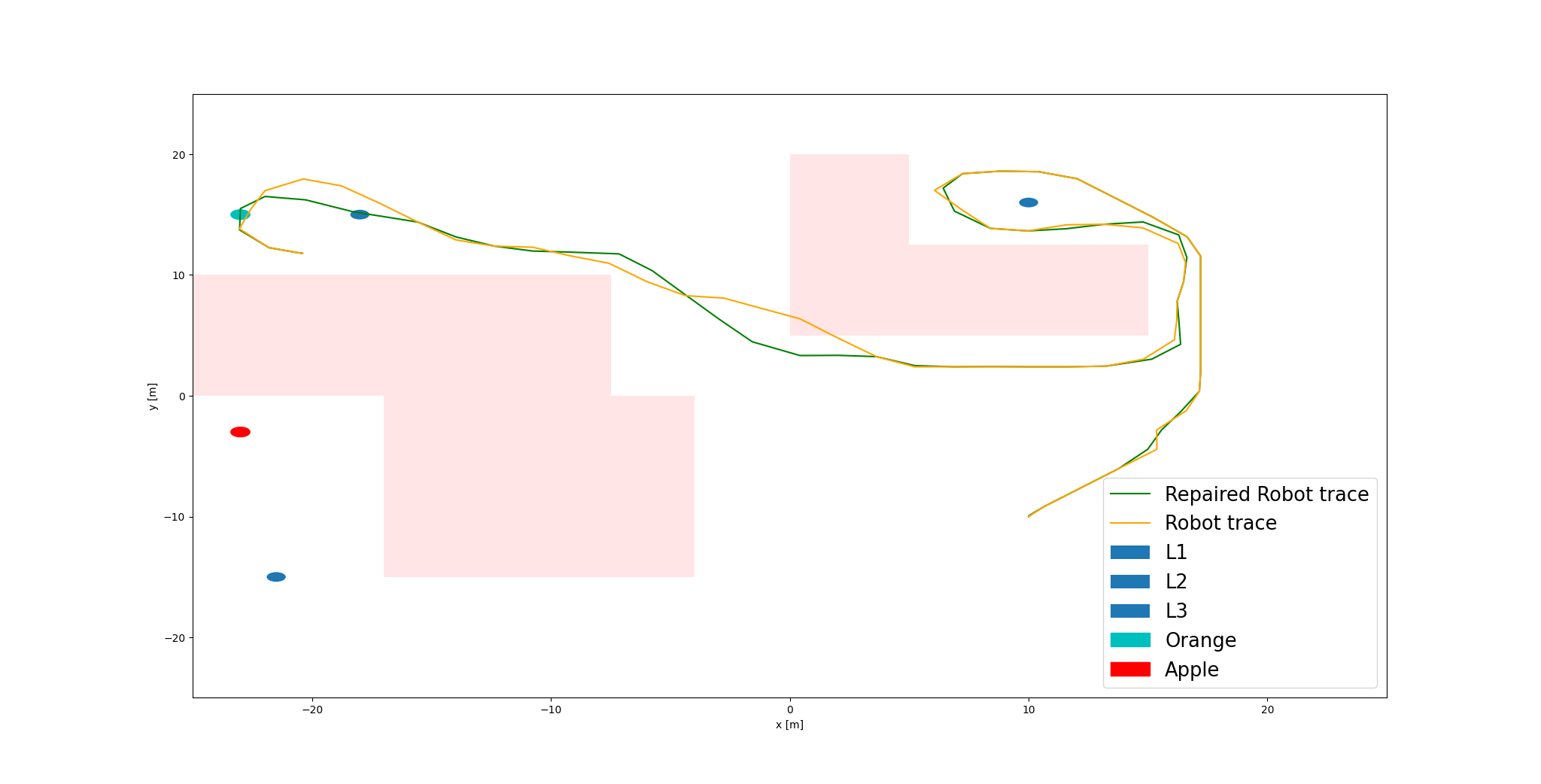}
\end{minipage}
\begin{minipage}{0.5\textwidth}
    \includegraphics[trim={4cm 2.0cm 4.5cm 3.0cm},clip, width=1.0\textwidth]{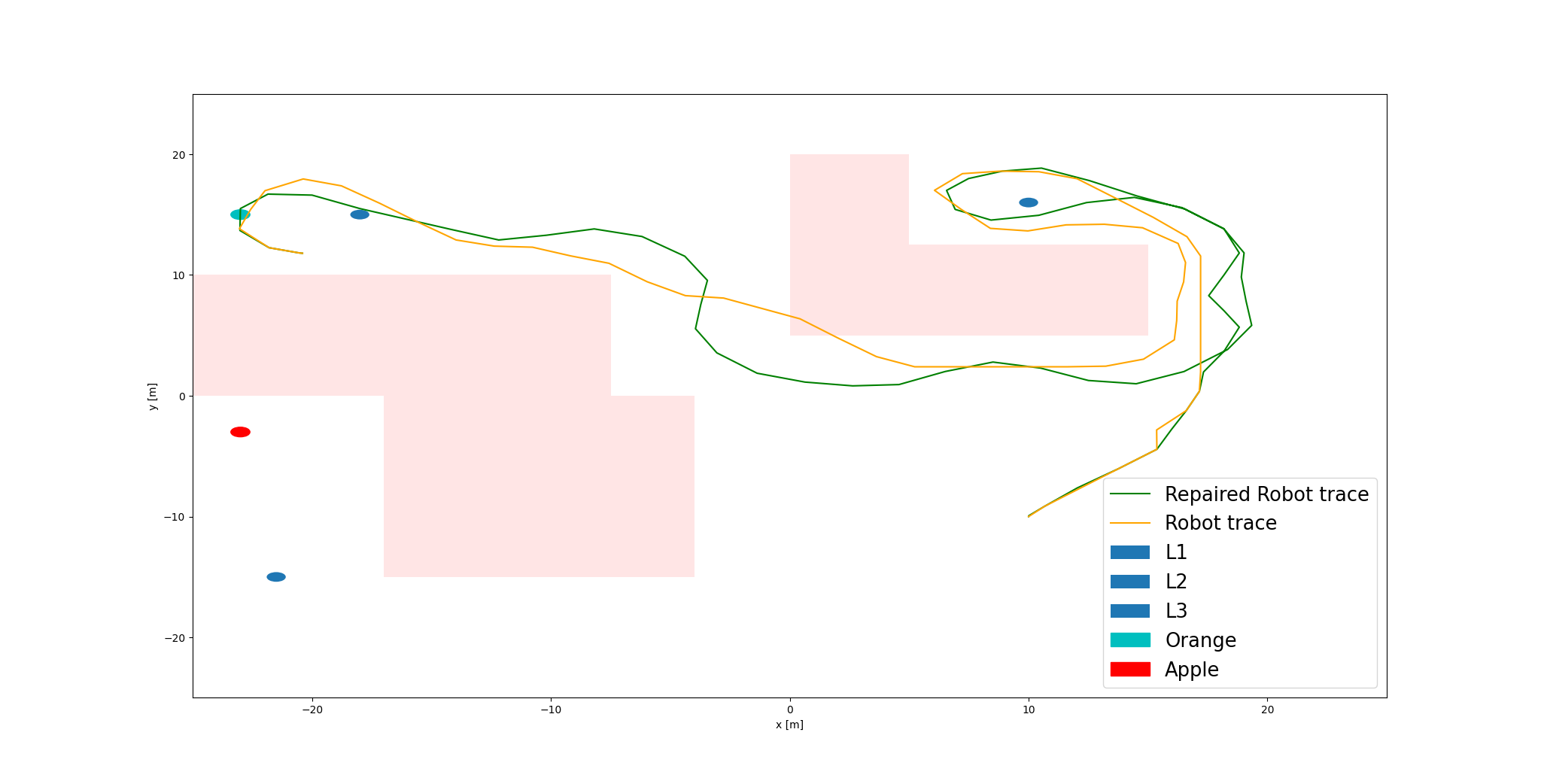}
\end{minipage}
\caption{Repair using different cost functions: on the left using $L1$ and on the right using combination $W$ with weights ($L1: 0.01, R: 0.99$).}
\label{fig:differentcost}
\end{figure}

\subsection{Automated Ship Deck Landing}
Landing on a ship deck is a  challenging task, wherein various landing aids and maneuvers need to be carefully selected~\cite{dlr140951,dlrreport}.
The benefits of TBT segmentation for an automated lander were previously discussed in \cite{10.1145/3641513.3650180}.
The depicted TBT in Figure \ref{fig:tbt} is a simplified version of the used TBT $\tree$, where all landing maneuvers are used in the fallback node: \emph{Straight-in}, \emph{Lateral}, \emph{45-Degree}, and \emph{Oblique}.
Each of them is represented by a sequence node.
Here, we consider the \emph{45-Degree} sequence node: $\Sequence([\Leaf(\F\ \nextTo(s)), \Leaf(\G_{[0,5s]}\ (\nextTo(s)$ $\wedge~ \alignedHeading(s)))$, $\Leaf(\alignedHeading(s)\ \Un\ \aboveTouchdown(s) \wedge \bullet)])$ where $s$ contains the position, velocity, and heading of the ship and the UAV.
We abbreviate this sequence by $\Sequence([\tree_1, \tree_2, \tree_3])$, e.g., $\tree_1=\Leaf(\F\ \nextTo(s))$.
The $\ap$ $\nextTo$ represents that the UAV is diagonally behind the ship, $\alignedHeading$ ensures that the UAV has a heading that is aligned with the ship heading, $\aboveTouchdown$ represents that the UAV is above the touchdown point of the ship.
The other sequence nodes beneath the fallback are similar, only the target position and the prescribed heading change.
The final leaf node also given in Figure \ref{fig:tbt} is common for all behaviors and specifies the descend property  $\Leaf(\F\ \descended(s))$, where $\descended(s)$ states that the UAV landed on the touchdown point.
We abbreviate this node by $\tree_4$.
For more detailed information on the TBT $\tree$, we refer to \cite{10.1145/3641513.3650180}.
Next, we repair a violating trace from \cite{10.1145/3641513.3650180} using segmentation information, showing the \emph{45-Degree} landing maneuver being the closes to satisfy the specification.

\noindent\emph{Incremental Repair}
All segmentations were computed in under $30$s.
The original trace had a mission-time of $126$s and a length of $25,349$, subsampled to $1014$ by the segmentation.
Nonlinear helicopter dynamics were simplified into four independent integrator chains similar to  Example \ref{exp:dynamic_system}: one for each of the three inertial axes ($x,y,z$) and one for the heading \cite{kooOutputTrackingControl1998}.
The repair works on the subsampled trace using $L1$ as the cost function.
The segments of the leaf nodes are: $(\tree_1, 0, 265)$, $(\tree_2, 266, 306)$, $(\tree_3, 307, 581)$, and $(\tree_4, 582, 1013)$.
Segments $\tree_1, \tree_2,$ and $\tree_3$ are violating.
As a reference for a full MILP encoding of the entire trace, we encoded a simplified landing as $\Leaf(\landing)$ using the formula $\landing= \F (\nextTo(s) \wedge \X (\stayPos \wedge \F_{[6,\infty]}( \mathit{\alignedHeading}(s) \wedge \aboveTouchdown(s) \wedge \F \descended(s) )))$, i.e., we replaced $\Un$ with $\F$ by omitting the left side.
We chose $\Leaf(\landing)$ because a full MILP encoding, i.e., directly solving the root node, caused an out-of-memory error.
Note that this allows a baseline comparison to standard approximations techniques supported by the optimizer.
We use Gurobi with its default parameters, which include features such as root relaxation and presolve. 
Experimental results in Table \ref{table:landing} show that the reference full MILP repair of $\Leaf(\landing)$ does not scale well for longer trace -- it took over 2000s -- while incremental repair took less than a minute.
The detailed steps of the incremental repair show that ensuring $\valid$ transitions did not save time for Steps 1, 4, and 6 ($8.87$s were unnecessary spend), but Step 3 saved presumably around $11$s.
The reason for this is that the last state of the segment must change too radically to satisfy the specification and only when both leaf nodes are encoded (Step 3), a matching state can be found.
The evaluation shows that \textsf{getCommonAncestor} (Line 18 in Algorithm \ref{alg:incremental}) was never invoked; only the affected leaves needed to be accounted for.
\begin{wrapfigure}{r}{0.65\textwidth}
  \vspace{-1cm}
  \begin{center}
    \includegraphics[width=0.65\textwidth]{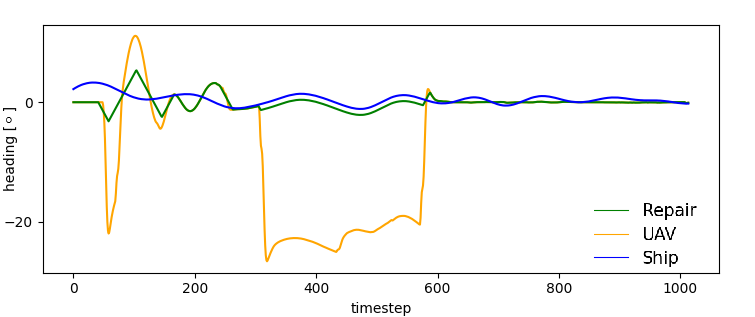}
  \end{center}
  \caption{Repair to satisfy the heading constrains in $\tree_3$.}
  \vspace{-0.5cm}
  \label{fig:L1}
\end{wrapfigure}
This avoided the need for a costly repair.
The result of the repair is shown in Figures \ref{fig:plotlanding} and \ref{fig:L1}.
The plots show that the segmentation provides a good way to decompose the specification as only small adjustments of the positions were required to satisfy $\tree$.
The most significant changes were due to $\tree_3$ that requires an aligned heading (see Timesteps 350 to 600 in Figure \ref{fig:L1}).

\begin{table*}
    \centering
    \begin{tabular}{c|c|c}
        \textbf{Formula} & \textbf{$L$} & \textbf{Time (s)}\\
        \hline
        \hline
        $\Leaf(\landing)$  & $\{\{ (\Leaf(\landing), 0, 1013) \}\}$ & 2200.40\\
        \hline
        $\tree$ & $\{\{(\tree_1, 0, 265)\}, \{(\tree_2, 266, 306)\}, \{(\tree_3, 307, 581)\}\}$ & \textbf{29.62}\\
        Step 1 & $l= \{(\tree_3, 307, 581)\}, f=\valid$ & \emph{3.72} \xmark\\
        Step 2 & $l= \{(\tree_3, 307, 581)\}, f=\loose$ & \emph{3.86} \cmark\\
         & $\{\{(\tree_1, 0, 265)\}, \{(\tree_2, 266, 306)\}, \{(\tree_3, 307, 581), (\tree_4, 582, 1013)\}\}$ \\
        Step 3 & $l= \{\{(\tree_3, 307, 581), (\tree_4, 582, 1013)\}\}, f=\valid$ & \emph{10.93} \cmark\\
         & $\{\{(\tree_1, 0, 265)\}, \{(\tree_2, 266, 306)\}\}$\\
        Step 4 & $l= \{(\tree_2, 266, 306)\}, f=\valid$ & \emph{0.60} \xmark\\
        Step 5 & $l= \{(\tree_2, 266, 306)\}, f=\loose$ & \emph{0.85} \cmark\\
         & $\{\{(\tree_1, 0, 265), (\tree_2, 266, 306)\}\}$\\
        Step 6 & $l= \{(\tree_1, 0, 265), (\tree_2, 266, 306)\}, f=\valid$ & \emph{4.55} \xmark\\
        Step 7 & $l= \{(\tree_1, 0, 265), (\tree_2, 266, 306)\}, f=\loose$ & \emph{5.11} \cmark\\
    \end{tabular}\vspace{3mm}
    \caption{Trace repair results with intermediate steps of the incremental repair. The Time column includes setting up the model and solving it. The results show that incrementally repairing $\tree$ is more efficient than the reference repair of $\Leaf(\landing)$. \xmark\ represents infeasible runs whereas \cmark\ represents successful runs.}
    \label{table:landing}
    \vspace{-0.5cm}
\end{table*}

\noindent\emph{Landmark-based Repair}
To illustrate the impact of disjunctions, we use incremental repair while omitting $\bullet$ from the last leaf node of the \emph{45-Degree} sequence node.
Specifically, we consider $\Leaf(\alignedHeading(s)\ \Un\ \aboveTouchdown(s))$ instead of $\Leaf(\alignedHeading(s)\ \Un\ \aboveTouchdown(s) \wedge \bullet)$. 
Therefore, the optimizer must determine the optimal position to satisfy $\aboveTouchdown(s)$, while $\bullet$ constraints it to the last position of this segment. 
As a result, the computation time increases from $29.62$s as in Table \ref{table:landing} to $278.32$s using incremental repair.
Figure~\ref{fig:landmark_experiments} illustrates this effect, comparing it to the results from the iterative landmark-based repair, where dots represent when solutions were found.
The time limit was set to $300$s.
The landmark-based repair finds its first solution after just $12$s and continues to improve upon it. 
Within approximately $40$s, a repair is achieved that is comparable to the one found by the incremental repair while saving $235$s.
\begin{wrapfigure}{r}{0.5\textwidth}
\vspace{-1.0cm}
  \begin{center}
    \begin{tikzpicture}
\begin{axis}[
    width=0.5\textwidth,
    axis lines = left,
    xlabel = \(Time (s)\),
    xtick distance=50,xmax=305,
    ylabel = {$J(\sigma, \sigma_R)$},
    ytick distance=200,ymin=5100,ymax=6520
]
\addplot [blue, mark=*] table[x index=1, y index=2] {plots/landmark_graph_without_bullet.txt};
\addlegendentry{Landmark-based}
\addplot [red, only marks, mark=*] coordinates {(278.32,5146.71)};
\addlegendentry{Incremental}

\node at (axis cs: 230, 5810) {$\Delta Time = 235$s};
\node at (axis cs: 250, 5650) {$\Delta J = 2.5$};
\addplot [thick,<-] coordinates {(50,5200) (170,5650)};
\addplot [thick,<-] coordinates {(270,5200) (175,5650)};

\end{axis}
\end{tikzpicture}
  \end{center}
\vspace{-0.7cm}
  \caption{Comparison of repair strategies.}
   \label{fig:landmark_experiments}
\vspace{-0.5cm}
\end{wrapfigure}
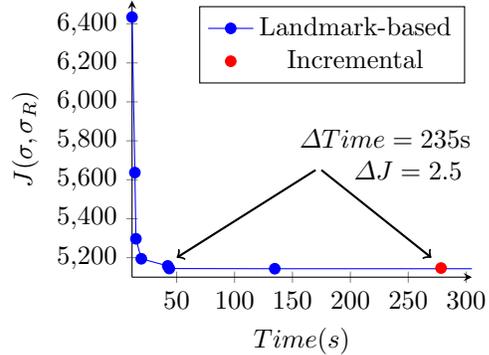
Figure~\ref{fig:landmark_plot} shows the repair. 
Note that, without the $\bullet$, the leaf node $\tree_3$, which contains the $\aboveTouchdown$ proposition, can be satisfied earlier (around Timestep $450$), thereby preventing the need of repairing heading thereafter.
Additionally, the repair chooses the same position to repair the $\nextTo$ proposition in $\tree_1$ as in Figure \ref{fig:plotlanding}.
As next experiment, we applied the same TBT specification, but instead of using the subsampled trace, we encoded the full original trace that contains $25,348$ entries. 
We were able to successfully identify a repair within $362$s.

\begin{figure}[t]
    \centering
    \includegraphics[width=0.6\textwidth]{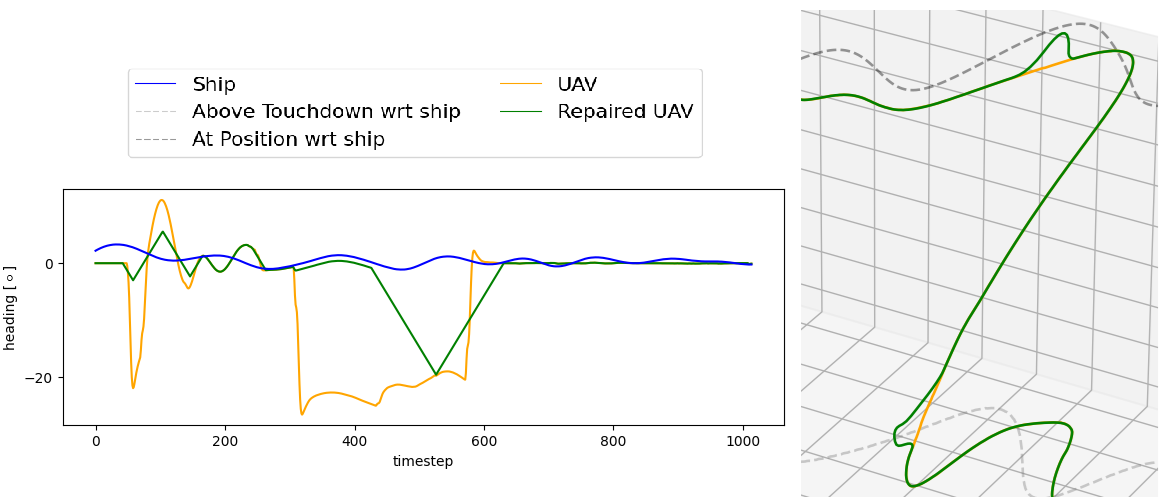}
    \caption{The repair chooses an earlier $\aboveTouchdown$ when omitting $\bullet$.}
    \label{fig:landmark_plot}
\end{figure}

\section{Related Work}
STL is mainly used as a specification language for controlling system behaviors.
In \cite{raman2014model}, STL formulas are transformed into a MILP for model-predictive control, generating control inputs that ensure conformance to the specification.
However, MILP's complexity limits this approach, especially with nested formulas or longer trajectories, making it unsuitable for long-horizon trajectory planning.
To address these limitations, \cite{kapoor2024safe} proposes a method to structurally decompose STL formulas to then incrementally solve them.
However, their method of decomposition does not handle disjunction.
In this work, rather than planning trajectories, we combine TBT segmentation and a MILP encoding of TBTs to repair given trajectories such that the repaired trajectory satisfies its TBT specification.

Landmarks have been studied for strategy solving \cite{DBLP:conf/cav/BaierCFFJS21} and planning \cite{DBLP:journals/corr/abs-1107-0052} as key features that must be true on any solution.
In this work, we adopt a similar conceptual role but with a distinct technical use.
In this paper, a landmark serves as a sufficient feature that guarantees the satisfaction of a TBT specification.
This enables a more efficient linear program encoding, significantly extending the capability to handle longer traces.

Falsification \cite{DBLP:conf/tacas/AnnpureddyLFS11} tries to find traces that ``falsify'' a given specification by using stochastic optimization to minimize its robustness.
Trace synthesis \cite{DBLP:conf/cav/SatoAZH24}, on the other hand, generates traces that satisfy the specification.
In contrast, this work addresses the problem of trace repair: given a violating trace, we minimally modify it so that the resulting trace satisfies the specification.
Our approach avoids stochastic optimization and uses a MILP formulation that ensures specification satisfaction while minimizing changes to the violating trace.

\section{Conclusions}
We have presented methods for repairing traces of CPS that violate a given TBT specification.
While a MILP could theoretically solve the problem, our experiments show that this is too expensive in practice.
To address this, we introduced an incremental repair strategy that uses the segmentation information from a TBT monitor to repair violating segments locally. 
Additionally, we presented a landmark-based repair strategy, an iterative approach that avoids MILP encoding of the TBT by using landmarks.
The landmarks allow us to formulate the repair as a linear program.
Our experiments demonstrate that the two strategies make it possible to repair traces of more than $25,000$ entries in under ten minutes, while the full MILP runs out of memory.
Future work will explore the use of trace repair for reinforcement learning, focusing on situations where agents fail their task and need assistance.

%
%
\bibliographystyle{splncs04}
\bibliography{references}

\end{document}